\documentclass{article}
\usepackage{amssymb}

\usepackage{etex}
\usepackage{amsmath}
\usepackage{tikz}
\usepackage{tkz-graph}
\usepackage{bussproofs}
\usepackage{graphicx}

\usetikzlibrary{arrows,automata}
\usetikzlibrary{arrows,shapes,positioning}
\usetikzlibrary{tikzmark,decorations.pathreplacing,calc}
\newtheorem{theorem}{Theorem}

\newtheorem{claim}[theorem]{Claim}
\newtheorem{conclusion}[theorem]{Conclusion}

\newtheorem{corollary}[theorem]{Corollary}

\newtheorem{definition}[theorem]{Definition}

\newtheorem{lemma}[theorem]{Lemma}

\newtheorem{remark}[theorem]{Remark}

\newtheorem{partial solution}[theorem]{Partial Solution}

\newenvironment{proof}[1][Proof]{\textbf{#1.} }{\ \rule{0.5em}{0.5em}}

\input{tcilatex}

\begin{document}

\begin{center}
{\Large On Proof Theory in Computer Science\medskip }

L. Gordeev, E. H. Haeusler

\textit{Universit\"{a}t T\"{u}bingen, Ghent University, PUC Rio de Janeiro}

l\texttt{ew.gordeew@uni-tuebingen.de}\textit{,}

\textit{PUC Rio de Janeiro}

\texttt{hermann@inf.puc-rio.br}
\end{center}

\section{Introduction}

The subject \emph{logic in computer science} should entail proof theoretic
applications. So the question arises whether open problems in computational
complexity can be solved by advanced proof theoretic techniques. In
particular, consider the complexity classes $\mathbf{NP}$, $\mathbf{coNP}$
and $\mathbf{PSPACE}$. It is well-known that $\mathbf{NP}$ and $\mathbf{coNP}
$ are contained in $\mathbf{PSPACE}$, but till recently precise
characterization of these relationships remained open. Now \cite{GH1}, \cite
{GH2} (see also \cite{GH3}) presented proofs of corresponding equalities $%
\mathbf{NP}=\mathbf{coNP}=\mathbf{PSPACE}$. These results were obtained by
appropriate proof theoretic tree-to-dag compressing techniques to be briefly
explained below. But let us first recall basic definitions of complexity
classes involved.

\subsection{Complexity classes}

Recall standard definitions of the complexity classes $\mathbf{NP}$, $%
\mathbf{coNP}$ and $\mathbf{PSPACE}$. A given language $L\subseteq \left\{
0,1\right\} ^{\ast }$\ is in $\mathbf{NP}$, resp. $\mathbf{coNP}$, if there
exists a polynomial $p$ and a polynomial-time TM $M$ such that for every $%
x\in \left\{ 0,1\right\} ^{\ast }$: 
\begin{equation*}
\begin{array}{c}
\quad \quad \ \fbox{$x\in L\Leftrightarrow \left( \exists u\in \left\{
0,1\right\} ^{p\left( \left| x\right| \right) }\right) \!M\left( x,u\right)
=1$}\qquad \qquad \qquad \left( \mathbf{NP}\right) \\ 
\quad \text{resp. }\fbox{$x\in L\Leftrightarrow \left( \forall u\in \left\{
0,1\right\} ^{p\left( \left| x\right| \right) }\right) $\negthinspace $%
\,M\left( x,u\right) =1$}\qquad \qquad \qquad \left( \mathbf{coNP}\right)
\end{array}
\end{equation*}
That is to say, a given $x\in \left\{ 0,1\right\} ^{\ast }$ is in $L$ iff $M$%
's execution on input $\left( x,u\right) $ provides output $1$ for some
(resp. every) $u\in \left\{ 0,1\right\} ^{\ast }$ of the length $\left|
u\right| \leq p\left( \left| x\right| \right) $, where $p$ and $M$ are
determined by $x$. Note that $\mathbf{coNP}$ is complementary to $\mathbf{NP}
$ (and vice versa), i.e. $L\in \mathbf{coNP}\Leftrightarrow L\notin \mathbf{%
NP}$. However it is unclear a priori whether\ symmetric difference $\left( 
\mathbf{NP\setminus coNP}\right) \mathbf{\cup }\left( \mathbf{coNP\setminus
NP}\right) $ is empty or not, as $card\left( \left\{ 0,1\right\} ^{p\left(
\left| x\right| \right) }\right) $ is exponential in $x$. In the former case
we'll have $\mathbf{NP=coNP}$, which seems more natural and/or plausible, as
it reflects an idea of logical equivalence between model theoretical (re: $%
\mathbf{NP}$) and proof theoretical (re: $\mathbf{coNP}$) interpretations of
non-deterministic polynomial-time computability.

Now $L\subseteq \left\{ 0,1\right\} ^{\ast }$\ is in $\mathbf{PSPACE}$ if
there exists a polynomial $p$ and a TM $M$ such that for every input $x\in
\left\{ 0,1\right\} ^{\ast }$, the total number of non-blank locations that
occur during $M$'s execution on $x$ is at most $p\left( \left| x\right|
\right) $, and $x\in L\Leftrightarrow M\left( x\right) =1$. Thus $\mathbf{%
PSPACE}$ requires polynomial upper bounds only on the space -- but not time
-- of entire computation. It is well-known (and not hard to prove) that $%
\mathbf{PSPACE=coPSPACE}$, while $\mathbf{NP}$ and $\mathbf{coNP}$ are
contained in $\mathbf{PSPACE}$. It is unclear a priori whether\ at least one
of $\mathbf{NP}$, $\mathbf{coNP}$ is a proper subclass of $\mathbf{PSPACE}$.
It is clear, however, that the assumption $\mathbf{NP=PSPACE}$ implies $%
\mathbf{NP=coNP}$ (via $\mathbf{PSPACE=coPSPACE}$).

\subsection{Logic and proof systems}

\emph{Classical propositional logic} provides natural interpretations of $%
\mathbf{NP}$ and $\mathbf{coNP}$. Namely, the well-known propositional
satisfiability and validity problems $SAT$ and $VAL$ are, respectively, $%
\mathbf{NP}$- and $\mathbf{coNP}$-complete. That is, an $L$ canonically
encoding the set of satisfiable (resp. valid) propositional formulas is
universal for the whole class $\mathbf{NP}$ (resp. $\mathbf{coNP}$).
However,\ classical proof systems usually correlated with $VAL$ are less
helpful for the comparison $\mathbf{NP}$ vs $\mathbf{coNP}$, as the size of
conventional proofs of tautologies $x$ use to be exponential in $\left|
x\right| $. It seems that \emph{minimal} and \emph{intuitionistic} \emph{%
propositional logics} \cite{Joh}, \cite{PraMa} provide us with more suitable
refinements . Recall that the minimal logic is determined by the axioms $%
\alpha \rightarrow \left( \beta \rightarrow \alpha \right) $, $\left( \alpha
\rightarrow \left( \beta \rightarrow \gamma \right) \right) \rightarrow
\left( \left( \alpha \rightarrow \beta \right) \rightarrow \left( \alpha
\rightarrow \gamma \right) \right) $\ and rule \emph{modus ponens} $\fbox{$%
\dfrac{\alpha \quad \quad \alpha \rightarrow \beta }{\beta }$}$ in standard
Hilbert-style formalism whose vocabulary includes propositional variables
and propositional connective `$\rightarrow $' ($\alpha $, $\beta $, $\gamma $%
, etc. denote corresponding formulas). The intuitionistic logic extends
minimal one by adding one propositional constant $\bot $ (falsity) and new
axiom $\bot \rightarrow \alpha $. It is well-known that there are
polynomial-size validity-preserving embeddings of formulas in classical into
intuitionistic and intuitionistic into minimal logic, respectively. Apart
from Hilbert-style formalism, proof systems for minimal and intuitionistic
logic include Gentzen-style \emph{sequent calculus} (SC) and
Gentzen-Prawitz-style \emph{natural deductions} (ND). Both admit two
well-known proof-optimization: \emph{cut elimination} in SC and \emph{%
normalization} in ND. Cut elimination approach provides sound and complete
systems of inferences without \emph{cut rule} that is equivalent to the
modus ponens. Inferences in the resulting \emph{cutfree} SC systems satisfy
a sort of \emph{subformula property} (: all premise formulas occur as
(sub)formulas in the conclusions), which enables better proof search
strategies. We can also assume that the heights of cutfree proofs
(derivations) are linear in the weights of conclusions, although such
constrain is not obvious for intuitionistic and/or minimal logic, see \cite
{Hud}, \cite{GH1}. In ND, the normalization allows to use just \emph{normal}
proofs that are known to satisfy \emph{weak subformula property} (: every
formula occurring in a maximal thread occurs as (sub)formula in the
conclusion), see \cite{Prawitz}. However, there are no polynomial upper
bounds on the heights of arbitrary normal ND.

These optimizations have been elaborated for standard tree-like versions of
SC and ND. Note that tree-like approach can't provide polynomial upper
bounds on the size of resulting proofs. To achieve this goal we formalize
another idea of \emph{horizontal compression}. That is, in a given tree-like
proof we wish to merge all nodes labeled with identical objects (sequents or
formulas) occurring on the same level so that in the compressed \emph{%
dag-like} proof every level will contain mutually different objects. In the
case of SC even the compressed polynomial-height dag-like proofs still would
be too large due to possibly exponential number of distinct sequents
occurring in it. Now consider a ``short''\ normal tree-like ND whose height
is polynomial in the weight of conclusion. By the weak subformula property
we observe that the total weight of distinct (sub)formulas occurring in it
is polynomial in the weight of conclusion. As ND proofs operate with single
formulas (not sequents!), compressing this tree-like ND proof will provide
us with a desired polynomial-weight dag-like deduction. However, such
compressed dag-like deduction requires a modified notion of provability, as
merging different occurrences of identical formulas appearing as conclusions
in the same level of deduction might require a new \emph{separation} rule $%
\left( S\right) $%
\begin{equation*}
\fbox{$\left( S\right) :\dfrac{\overset{n\ times}{\overbrace{\alpha \quad
\cdots \quad \alpha }}}{\alpha \ }\ $($n$ arbitrary) }
\end{equation*}
whose identical premises are understood disjunctively: ``\emph{if at least
one premise is proved then so is the conclusion}'' (in contrast to ordinary
inferences: ``\emph{if all premises are proved then so are the conclusions}%
''). The notion of provability is modified accordingly such that proofs are
locally correct deductions assigned with appropriate sets of closed threads
that satisfy special conditions of \emph{local coherency} (in contrast to
ordinary local correctness, the local coherency is not verifiable in
polynomial time). These locally coherent threads are inherited by the
underlying closed tree-like threads. The required ``small''
polynomial-weight proof now arises by collapsing $\left( S\right) $ to plain
repetitions 
\begin{equation*}
\fbox{$\left( R\right) :\dfrac{\alpha }{\alpha \ }$}
\end{equation*}
with respect to the appropriately chosen premises of $\left( S\right) $. The
choice is made non-deterministically using the set of locally coherent
threads in question.

Keeping this in mind consider the $\mathbf{NP}$-complete Hamiltonian graph
problem and let $\rho $\ be a purely implicational formula expressing in
standard form that a given (simple and directed) graph $G$ has no
Hamiltonian cycles. We observe that the canonical tree-like proof search for 
$\rho $ in the minimal ND with standard inferences 
\begin{equation*}
\fbox{$\left( \rightarrow I\right) :\dfrac{\QDATOP{\QDATOP{\left[ \alpha %
\right] }{\vdots }}{\beta }}{\alpha \rightarrow \beta }$}\quad \fbox{$\left(
\rightarrow E\right) :\dfrac{\alpha \quad \alpha \rightarrow \beta }{\beta \ 
}$}
\end{equation*}
yields a normal tree-like proof $\partial $ whose height is polynomial in $%
\left| G\right| $ (and hence $\left| \rho \right| $), provided that $G$ is
non-Hamiltonian. Since $\partial $ is normal, it will obey the requested
polynomial upper bounds in question, and hence the weight of its dag-like
compression will be polynomially bounded, as desired. Summing up, for any
given non-Hamiltonian graph $G$ there is some polynomial-weight dag-like ND
refutation of the existence of Hamiltonian cycles in $G$. Note that
polynomial-weight ND proofs (tree- or dag-like) have polynomial-time
certificates (\cite{GH2}: Appendix), while the non-hamiltoniancy of simple
and directed graphs is $\mathbf{coNP}$-complete. Hence $\mathbf{coNP}$ is in 
$\mathbf{NP}$, which yields $\mathbf{NP=coNP}$.

To handle our main assertion $\mathbf{NP=PSPACE}$ we recall that the
validity of the minimal logic under consideration is known to be $\mathbf{%
PSPACE}$-complete. Moreover, every minimal tautology is provable in
Hudelmaier's cutfree SC (abbr.: HSC) for minimal logic \cite{Hud} by a
tree-like derivation whose height is linear in the weight of conclusion.
Furthermore, straightforward ND interpretation of such tree-like input in
HSC yields corresponding ``short'' (though not necessarily normal) tree-like
proof in ND for minimal logic whose total weight of distinct (sub)formulas
is polynomial in the weight of conclusion \cite{GH1}. Now the latter
tree-like ND proof is horizontally compressible to a polynomial-weight
dag-like proof by the same method as sketched above with respect to
``short'' normal ND. This yields $\mathbf{NP=PSPACE}$. A more detailed
presentation is as follows.

\section{Survey of proofs}

\subsection{Basic tree-like and dag-like ND}

Our basic ND calculus for minimal logic, \textsc{NM}$_{\rightarrow }$,
includes two basic inferences 
\begin{equation*}
\fbox{$\left( \rightarrow I\right) :\dfrac{\QDATOP{\QDATOP{\left[ \alpha %
\right] }{\vdots }}{\beta }}{\alpha \rightarrow \beta }$}\ ,\quad \fbox{$\
\left( \rightarrow E\right) :\dfrac{\alpha \quad \quad \alpha \rightarrow
\beta }{\beta }$}
\end{equation*}
and one auxiliary repetition rule $\fbox{$\left( R\right) :\dfrac{\alpha }{%
\alpha \ }$}$, where $\left[ \alpha \right] $ in $\left( \rightarrow
I\right) $\ indicates that all $\alpha $-leaves occurring above $\beta $%
-node exposed are \emph{discharged} assumptions (cf. \cite{Prawitz}). In 
\textsc{NM}$_{\rightarrow }$, tree-like deductions are understood as finite
rooted at most binary-branching trees whose nodes are labeled with purely
implicational formulas ($\alpha $, $\beta $, $\gamma $, etc.) that are
ordered according to the inferences exposed, as usual in proof theory,
whereas dag-like deductions are the analogous finite rooted dags. Thus for
any node $x$ in a tree-like deduction $\partial $, the set of all nodes
occurring below $x$ in $\partial $ is linearly ordered. This constrain is
lacking in dag-like deductions. Note that in dag-like \textsc{NM}$%
_{\rightarrow }$ deductions, all nodes can have at most two premises (=
children), but arbitrary many conclusions (= parents) \footnote{%
{\footnotesize This follows from standard conditions of the local
correctness. }}, whereas the latter is forbidden in the tree-like case.

\begin{definition}
A given (whether tree- or dag-like) \textsc{NM}$_{\rightarrow }$-deduction $%
\partial $ \emph{proves} its root-formula $\rho $ (abbr.: $\partial \vdash
\rho $) iff every maximal thread connecting the root with a leaf labeled $%
\alpha $ is closed (= discharged), i.e. it contains a $\left( \rightarrow
I\right) $ with conclusion $\alpha \rightarrow \beta $, for some $\beta $. A
purely implicational formula $\rho $ is \emph{valid in minimal logic} iff
there exists a tree-like \textsc{NM}$_{\rightarrow }$-deduction $\partial $
that proves $\rho \,$; such $\partial $ is called a \emph{proof} of $\rho $.
\end{definition}

\begin{remark}
Tree-like constraint in the definition of validity is inessential.

That is, for any dag-like $\partial \in \,$\textsc{NM}$_{\rightarrow }$ with
root-formula $\rho $, if $\partial \vdash \rho $ then $\rho $ is valid in
minimal logic. Because any given dag-like $\partial $ can be unfolded into a
tree-like deduction $\partial ^{\prime }$ by straightforward
thread-preserving top down recursion. To this end every node $x\in \partial $
with $n>1$\ distinct conclusions should be replaced by $n$ distinct but
identically labeled nodes $x_{1},\cdots ,x_{n}\in \partial ^{\prime }$ to be
connected with corresponding single conclusions. This operation obviously
preserves the closure of threads, i.e. $\partial \vdash \rho $ infers $%
\partial ^{\prime }\vdash \rho $.
\end{remark}

Formal verification of the assertion $\partial \vdash \rho $ is simple, as
follows -- whether for tree-like or, generally, dag-like $\partial $. Every
node $x\in \partial $ is assigned, by top-down recursion, a set of
assumptions $A\left( x\right) $ such that:

\begin{enumerate}
\item  $A\left( x\right) :=\left\{ \alpha \right\} $ if $x$ is a leaf
labeled $\alpha $,

\item  $A\left( x\right) :=A\left( y\right) $ if $x$ is the conclusion of $%
\left( R\right) $ with premise $y$,

\item  $A\left( x\right) :=A\left( y\right) \setminus \left\{ \alpha
\right\} $ if $x$ is the conclusion of $\left( \rightarrow I\right) $ with
label $\alpha \rightarrow \beta $\ and premise $y$,

\item  $A\left( x\right) :=A\left( y\right) \cup A\left( z\right) $ if $x$
is the conclusion of $\left( \rightarrow E\right) $ with premises $y,$ $z$.
\end{enumerate}

This easily yields

\begin{lemma}
Let $\partial \in \,$\textsc{NM}$_{\rightarrow }$ (whether tree- or
dag-like). Then $\partial \vdash \rho \Leftrightarrow A\left( r\right)
=\emptyset $ holds with respect to standard set-theoretic interpretations of
``$\,\cup $'' and ``$\,\setminus $'' in $A\left( r\right) $, where $r$ and $%
\rho $\ are the root and the root-formula of $\partial $, respectively.
Moreover, $A\left( r\right) \overset{?}{=}\emptyset $ is verifiable by a
deterministic TM in $\left| \partial \right| $-polynomial time, where by $%
\left| \partial \right| $ we denote the weight of $\partial $.
\end{lemma}

\begin{proof}
The equivalence easily follows by induction on the height of $\partial $.
The second assertion is completely analogous to the well-known
polynomial-time decidability of the circuit value problem.
\end{proof}

\begin{definition}
Tree-like\emph{\ }\textsc{NM}$_{\rightarrow }$-deduction $\partial $ with
the root-formula $\rho $ is called \emph{polynomial}, resp. \emph{%
quasi-polynomial}, if\ its weight (= total number of symbols), resp. height
plus total weight of distinct formulas,\ is polynomial in the weight of
conclusion, $\left| \rho \right| $.
\end{definition}

\begin{theorem}
Any quasi-polynomial tree-like proof $\partial \vdash \rho $ can be
compressed into a polynomial dag-like proof $\partial ^{\ast }\vdash \rho $.
\end{theorem}

The mapping $\partial \hookrightarrow \partial ^{\ast }$ is obtained by a
two-folded horizontal compression\ $\partial \hookrightarrow \partial
^{\flat }\hookrightarrow \partial ^{\ast }$, where $\partial ^{\flat }$ is
dag-like deduction in the following modified ND that extends \textsc{NM}$%
_{\rightarrow }$ by the \emph{separation} rule $\left( S\right) $, cf.
Introduction.

\subsection{ND with the separation rule}

Recall that the \emph{separation} rule $\left( S\right) $%
\begin{equation*}
\fbox{$\left( S\right) :\dfrac{\overset{n\ times}{\overbrace{\alpha \quad
\cdots \quad \alpha }}}{\alpha \ }\ $($n$ arbitrary) }
\end{equation*}
is understood disjunctively: ``\emph{if at least one premise is proved then
so is the conclusion}'' (in contrast to ordinary inferences: ``\emph{if all
premises are proved then so are the conclusions}''). Let \textsc{NM}$%
_{\rightarrow }^{\flat }$ extend \textsc{NM}$_{\rightarrow }$ by adding a
new inference $\left( S\right) $. The notion of provability in \textsc{NM}$%
_{\rightarrow }^{\flat }$\ is modified as follows. To begin with, for any 
\textsc{NM}$_{\rightarrow }^{\flat }$\ deduction $\partial $ we modify our
basic definition of the set of assignments $\left\{ A\left( x\right) :x\in
\partial \right\} $ by adding to old recursive clauses 1--4 (see above) a
new clause 5 with new separation symbol $\circledS $ :

\begin{description}
\item  5. $A\left( x\right) =\circledS \left( A\left( y_{1}\right) ,\cdots
,\,A\left( y_{n}\right) \right) $ if $x$ is the conclusion of $\left(
S\right) $ with premises $y_{1},\cdots ,y_{n}$.
\end{description}

Having this done we stipulate

\begin{definition}
For any given (whether tree- or dag-like) deduction $\partial \in \,$\textsc{%
NM}$_{\rightarrow }^{\flat }$\ with root $r$ and root-formula $\rho $, $%
\partial $ is called a \emph{modified proof} of $\rho $ (abbr.: $\partial
\vdash ^{\flat }\rho $) if $A\left( r\right) $ reduces to $\emptyset $
(abbr.: $A\left( r\right) \vartriangleright \emptyset $) by standard
set-theoretic interpretations of ``$\,\cup $'', ``$\,\setminus $'' and
nondeterministic disjunctive valuations $\circledS \left( t_{1},\cdots
,\,t_{n}\right) :=t_{i}$, for any chosen $i\in \left\{ 1,\cdots ,n\right\} $%
. Obviously $\partial \vdash ^{\flat }\rho \Leftrightarrow \partial \vdash
\rho $ holds for every separation-free $\partial $.
\end{definition}

\begin{lemma}
For any $\partial $ as above, if $\partial \vdash ^{\flat }\rho $\ then $%
\rho $ is valid in minimal logic. Moreover, the assertion $\partial \vdash
^{\flat }\rho $ can be confirmed by a nondeterministic TM in $\left|
\partial \right| $-polynomial time.
\end{lemma}

\begin{proof}
The former assertion reduces to its trivial \textsc{NM}$_{\rightarrow }$
case (see above). For suppose that $A\left( r\right) \vartriangleright
\emptyset $ holds with respect to a successive nondeterministic valuation of
the occurrences $\circledS $. This reduction determines a successive
bottom-up thinning of $\partial $ that results in a ``cleansed'' $\left(
S\right) $-free subdeduction $\partial _{0}\in \,$\textsc{NM}$_{\rightarrow
}^{\flat }$. Thus $A\left( r\right) \vartriangleright \emptyset $ in $%
\partial $\ implies $A\left( r\right) =\emptyset $ in $\partial _{0}$. Since 
$\left( S\right) $ does not occur in $\partial _{0}$ anymore, we have $%
\partial _{0}\in \,$\textsc{NM}$_{\rightarrow }$, and hence $\partial
_{0}\vdash \rho $ holds by Lemma 3. So by previous considerations with
regard to \textsc{NM}$_{\rightarrow }$ we conclude that $\rho $ is valid in
minimal logic, which can be confirmed\ in $\left| \partial \right| $%
-polynomial time, as required (see Lemma 3).
\end{proof}

\subsection{Horizontal compression with cleansing}

In the sequel for any natural deduction $\partial $ we denote by $h\left(
\partial \right) $ and $\phi \left( \partial \right) $ the height of $%
\partial $ and the total weight of the set of distinct formulas occurring in 
$\partial $, respectively. Now we are prepared to explain proof of Theorem
5. For any tree-like \textsc{NM}$_{\rightarrow }$ proof $\partial $ of $\rho
\,$ let $\partial ^{\prime }\in \,$\textsc{NM}$_{\rightarrow }$ be its
horizontal compression defined by bottom-up recursion on $h\left( \partial
\right) $ such that for any $n\leq h\left( \partial \right) $, the $n^{th}$\
horizontal section of $\partial ^{\flat }$ is obtained by merging all nodes
with identical formulas occurring in the $n^{th}$\ horizontal section of $%
\partial $. The inferences in $\partial ^{\prime }$ are naturally inherited
by the ones in $\partial $. Obviously $\partial ^{\prime }$ is a dag-like
(not necessarily tree-like anymore) deduction with the root formula $\rho $.
However, $\partial ^{\prime }$ need not preserve the local correctness with
respect to basic inferences $\left( \rightarrow I\right) $, $\left(
\rightarrow E\right) $, $\left( R\right) $. For example, a compressed
multipremise configuration 
\begin{equation*}
\fbox{$\ \left( \rightarrow I,E\right) :\dfrac{\beta \quad \quad \gamma
\quad \quad \gamma \rightarrow \left( \alpha \rightarrow \beta \right) }{%
\alpha \rightarrow \beta }$}
\end{equation*}
that is obtained by merging identical conclusions $\alpha \rightarrow \beta $
of 
\begin{equation*}
\fbox{$\ \left( \rightarrow I\right) :\dfrac{\beta }{\alpha \rightarrow
\beta }$}\quad \text{and\quad }\fbox{$\ \left( \rightarrow E\right) :\dfrac{%
\gamma \quad \quad \gamma \rightarrow \left( \alpha \rightarrow \beta
\right) }{\alpha \rightarrow \beta }$}
\end{equation*}
is not a correct inference in \textsc{NM}$_{\rightarrow }$. To overcome this
trouble we upgrade $\partial ^{\prime }$ to a modified deduction $\partial
^{\flat }$ that separates such multiple premises using appropriate instances
of the separation rule $\left( S\right) $. For example, $\left( \rightarrow
I,E\right) $ as above should be replaced by this \textsc{NM}$_{\rightarrow
}^{\flat }$-correct configuration 
\begin{equation*}
\fbox{$\left( S\right) :\dfrac{\ \left( \rightarrow I\right) :\ \dfrac{\beta 
}{\alpha \rightarrow \beta \ }\quad \left( \rightarrow E\right) \ :\dfrac{%
\gamma \quad \quad \gamma \rightarrow \left( \alpha \rightarrow \beta
\right) }{\alpha \rightarrow \beta \ }}{\ \alpha \rightarrow \beta \ }$}%
\text{.}
\end{equation*}
This $\partial ^{\flat }$ is a locally correct dag-like (not necessarily
tree-like anymore) deduction in \textsc{NM}$_{\rightarrow }^{\flat }$ with
the root formula $\rho $. Moreover $\partial ^{\flat }$ is polynomial as $%
\left| \partial ^{\flat }\right| \leq 2\left| \partial ^{\prime }\right| $
and $\left| \partial ^{\prime }\right| \leq h\left( \partial \right) \times
\phi \left( \partial \right) $. However, we can't claim that $\partial
^{\flat }$ proves $\rho $ because arbitrary maximal dag-like threads in $%
\partial ^{\flat }$ can arise by concatenating different segments of
different threads in $\partial $, which can destroy the required closure
condition\ (cf. Definition 1). On the other hand, we know that all threads
in $\partial $ are closed, so let $\mathcal{F}^{\flat }$ be the dag-like
image in $\partial ^{\flat }$ of these tree-like threads under the mapping $%
\partial \hookrightarrow \partial ^{\flat }$. We observe that $\mathcal{F}%
^{\flat }$ satisfies the following three conditions of \emph{local coherency}%
, where $n:=h\left( \partial ^{\flat }\right) $ and for any (maximal
bottom-up) thread $\Theta =\left[ r=x_{0},\cdots ,x_{n}\right] \in \mathcal{F%
}^{\flat }$ and $i\leq n$ we let $\Theta \!\upharpoonright _{x_{i}}:=\left[
x_{0},\cdots ,x_{i}\right] $.

\begin{enumerate}
\item  $\mathcal{F}^{\flat }$ is dense in $\partial ^{\flat }$, i.e. $\left(
\forall u\in \partial ^{\flat }\right) \left( \exists \Theta \in \mathcal{F}%
^{\flat }\right) \left( u\in \Theta \right) $.

\item  Every $\Theta \in \mathcal{F}^{\flat }$ $\,$is closed, i.e. its
leaf-formula $\alpha \!\left( x_{n}\right) $ is discharged in $\Theta $.

\item  $\mathcal{F}$ preserves $\left( \rightarrow E\right) $, i.e.

$\left. 
\begin{array}{c}
\left( \forall \Theta \in \mathcal{F}^{\flat }\right) \left( \forall u\in
\Theta \right) \left( \forall v\neq w\in Child_{\partial ^{\flat }}\left(
u\right) :\ v\in \Theta \right) \qquad \qquad \qquad \ \ \  \\ 
\left( \exists \Theta ^{\prime }\in \mathcal{F}^{\flat }\right) \left( w\in
\Theta ^{\prime }\wedge \Theta \!\upharpoonright _{u}=\Theta ^{\prime
}\!\upharpoonright _{u}\right) .\qquad \quad \qquad \qquad \qquad .\qquad
\qquad \quad
\end{array}
\right. $
\end{enumerate}

In the sequel we call any $\mathcal{F}^{\flat }$ satisfying conditions of
local coherency the \emph{fundamental set of threads} (abbr.: \emph{fst}) in 
$\partial ^{\flat }$.

\begin{lemma}
Let $\partial ^{\flat }$ be any given locally correct dag-like \textsc{NM}$%
_{\rightarrow }^{\flat }$-deduction with root-formula $\rho $ that is
supplied with a fst $\mathcal{F}^{\flat }$. Then $\rho $ has a modified
dag-like proof $\partial ^{\ast }\subseteq \partial ^{\flat }$.
\end{lemma}

\begin{proof}
We show that $\mathcal{F}^{\mathrm{\flat }}$ determines successive
left-to-right $\circledS $-eliminations $\circledS \left( \!A\left(
y_{1}\right) ,\cdots ,\,A\left( y_{n}\right) \!\right) \!\hookrightarrow
\!A\left( y_{i}\right) $ inside $A\left( r\right) $ leading to a desired
reduction $A\left( r\right) \vartriangleright \emptyset $ (see basic
notations in 2.2). These eliminations together with a suitable sub-\emph{fst}
$\mathcal{F}_{0}^{\mathrm{\flat }}$ $\subseteq \mathcal{F}^{\mathrm{\flat }}$
arise as follows\ by bottom-up recursion along $\mathcal{F}^{\mathrm{\flat }%
} $. Let $x$ be a chosen lowest conclusion of $\left( \rightarrow E\right) $
in $\partial ^{\flat }$, if any exists. By the density of $\mathcal{F}^{%
\mathrm{\flat }}$, there exists $\Theta \in \mathcal{F}^{\mathrm{\flat }}$
with $x\in \Theta $; so let $\Theta \in \mathcal{F}_{0}^{\mathrm{\flat }}$.
Let $y$ and $z$\ be the two children of $x$ and suppose that $y\in \Theta $.
By the third, $\left( \rightarrow E\right) $-preserving \emph{fst} condition
there exists a $\Theta ^{\prime }\in \mathcal{F}^{\mathrm{\flat }}$ with $%
z\in \Theta ^{\prime }$ and $\Theta \!\upharpoonright _{x}=$ $\Theta
^{\prime }\!\upharpoonright _{x}$; so let $\Theta ^{\prime }\in \mathcal{F}%
_{0}^{\mathrm{\flat }}$ be the corresponding ``upgrade''of $\Theta $. If $%
z\in \Theta $ then let $\Theta ^{\prime }:=\Theta $. Note that $\Theta
\!\!\upharpoonright _{x}$ determines substitutions $A\left( u\right)
=\circledS \left( A\left( v_{1}\right) ,\cdots ,\,A\left( v_{n}\right)
\right) :=A\left( v_{i}\right) $ in all parents of the $\left( S\right) $%
-conclusions $u$ occurring in both $\Theta $ and $\Theta ^{\prime }$\ below $%
x$, if any exist, and thereby all $\circledS $-eliminations $A\left(
u\right) \hookrightarrow A\left( v_{i}\right) $ in the corresponding
subterms of $A\left( r\right) $. The same procedure is applied to the nodes
occurring in $\Theta $ and $\Theta ^{\prime }$\ between $x$ and the next
lowest conclusions of $\left( \rightarrow E\right) $; this yields new closed
threads $\Theta ^{\prime \prime },\Theta ^{\prime \prime \prime },\cdots \in 
\mathcal{F}_{0}^{\mathrm{\flat }}\subseteq \mathcal{F}^{\flat }$ and $%
\circledS $-eliminations in the corresponding initial fragments of $A\left(
r\right) $. We keep doing this recursively until the list of remaining $%
\circledS $-occurrences in $\Theta \in \mathcal{F}_{0}^{\mathrm{\flat }}$ is
empty. The final ``cleansed'' $\circledS $-free conversion of $A\left(
r\right) $ is represented by a set of formulas that easily reduces to $%
\emptyset $ by ordinary set-theoretic interpretation of the remaining
operations ``$\,\cup \,$'' and ``$\,\setminus "$, since every $\Theta \in 
\mathcal{F}_{0}^{\mathrm{\flat }}$ involved is closed. The correlated
``cleansed'' deduction $\partial ^{\ast }$ obtained by substituting
corresponding instances of $\left( R\right) $ for thus eliminated $\left(
S\right) $ is a locally correct dag-like deduction of $\rho $ in the $\left(
S\right) $-free fragment of \textsc{NM}$_{\rightarrow }^{\mathrm{\flat }}$,
and hence it belongs to \textsc{NM}$_{\rightarrow }$. Moreover the set of
maximal threads in $\partial ^{\ast }$ is uniquely determined by the
remaining rules $\left( R\right) $, $\left( \rightarrow I\right) $, $\left(
\rightarrow E\right) $. By the definition these ``cleansed'' maximal threads
are all included in $\mathcal{F}^{\flat }$ thus being closed with respect to 
$\left( \rightarrow I\right) $. \footnote{{\footnotesize These threads may
be exponential in number, but our nondeterministic algorithm runs on the
polynomial set of nodes. }} This yields a desired reduction $A\left(
r\right) \vartriangleright \emptyset $, i.e. $A\left( r\right) =\emptyset $,
in $\partial ^{\ast }$. Hence $\partial ^{\ast }$ proves $\rho $ in \textsc{%
NM}$_{\rightarrow }$. Obviously $\partial ^{\ast }$ is a subdeduction of $%
\partial ^{\flat }$.\medskip
\end{proof}

Operation $\partial ^{\flat }\hookrightarrow \partial ^{\ast }$ is also
referred to as \emph{horizontal cleansing}.

\begin{corollary}
By Lemma 7, the assertion of the lemma implies that $\rho $ is valid in
minimal logic. Actually $\partial ^{\ast }$ involved is separation-free,
which yields $\partial ^{\ast }\vdash \rho $.
\end{corollary}

This completes our proof of Theorem 5 and together with Lemma 3 yields

\begin{corollary}
Any given $\rho $ is valid in minimal logic iff there exists a polynomial
dag-like proof $\partial ^{\ast }$ of $\rho $, in \textsc{NM}$_{\rightarrow
} $. Moreover, the assertion $\partial ^{\ast }\vdash \rho $ can be
confirmed by a deterministic TM in $\left| \rho \right| $-polynomial time.
\end{corollary}

\subsection{Consequences for computational complexity}

\subsubsection{Case $\mathbf{NP}$ vs $\mathbf{coNP}$}

Since normal ND proofs satisfy weak subformula property, we have

\begin{lemma}
Any given normal\ tree-like \textsc{NM}$_{\rightarrow }$-proof $\partial $ of $%
\rho $ whose height $h\left( \partial \right) $ is polynomial in $\left|
\rho \right| $ is quasi-polynomial.
\end{lemma}

Let $P$ be a chosen NP-complete problem and purely implicational formula $%
\rho $ be valid iff $P$ has no positive solution. In particular, let $P$ be
the Hamiltonian graph problem and $\rho $ express in standard way that a
given graph $G$ has no Hamiltonian cycles. Suppose that the canonical proof
search of $\rho $ in \textsc{NM}$_{\rightarrow }$ yields a normal tree-like
proof $\partial $ whose height is polynomial in $\left| G\right| $ (and
hence $\left| \rho \right| $), provided that $G$ is non-Hamiltonian. Then by
the last lemma $\partial $ will be polynomially bounded. That is, we argue
as follows.

\begin{lemma}
Let $P$ be the Hamiltonian graph problem and purely implicational formula $%
\rho $ express that a given graph $G$ has no Hamiltonian cycles. There
exists a normal tree-like \textsc{NM}$_{\rightarrow }$-proof of $\rho $ such
that $h\left( \partial \right) $ is polynomial in $\left| G\right| $ (and
hence $\left| \rho \right| $), provided that $G$ is non-Hamiltonian.
\end{lemma}

Recall that polynomial ND proofs (whether tree- or dag-like) have
polynomial-time certificates, while the non-hamiltoniancy of simple and
directed graphs is $\mathbf{coNP}$-complete. Hence Corollary 10 yields

\begin{corollary}
$\mathbf{NP=coNP}$\textbf{\ }holds true.
\end{corollary}

So it remains to prove Lemma 12. To this end, consider a simple \footnote{%
{\footnotesize A simple graph has no multiple edges. For every pair of nodes 
}$(v_{1},v_{2})${\footnotesize \ in the graph there is at most one edge from 
}$v_{1}${\footnotesize \ to }$v_{2}$.} directed graph $G=\langle
V_{G},E_{G}\rangle $, $card\left( V_{G}\right) =n$. A \emph{Hamiltonian path}
(or \emph{cycle}) in $G$ is a sequence of nodes $\mathcal{X}%
=v_{1}v_{2}\ldots v_{n}$, such that, the mapping $i\mapsto v_{i}$ is a
bijection of $\left[ n\right] =\{1,\cdots ,n\}$ onto $V_{G}$ and for every $%
0<i<n$ there exists an edge $(v_{i},v_{i+1})\in E_{G}$. The (decision)
problem whether or not there is a Hamiltonian path in $G$ is known to be
NP-complete (cf. e.g. \cite{arora}). If the answer is YES then $G$ is called
Hamiltonian. In order to verify that a given sequence of nodes $\mathcal{X}$%
, as above, is a Hamiltonian path it will suffice to confirm that:

\begin{enumerate}
\item  There are no repeated nodes in $\mathcal{X}$,

\item  No element $v\in V_{G}$ is missing in $\mathcal{X}$,

\item  For each pair $\left\langle v_{i}v_{j}\right\rangle $ in $\mathcal{X}$
there is an edge $(v_{i},v_{j})\in E_{G}$.
\end{enumerate}

It is readily seen that the conjunction of $1,2,3$ is verifiable by a
deterministic TM in $n$-polynomial time. Consider a natural formalization of
these conditions (cf. e.g. \cite{arora}) in propositional logic with one
constant $\bot $ (\emph{falsum}) and three connectives `$\wedge $', `$\vee $%
', `$\rightarrow $'.

\begin{definition}
For any $G=\langle V_{G},E_{G}\rangle $, $card(V_{G})=n>0$, as above,
consider propositional variables $X_{i,v}$, $i\in \left[ n\right] $, $v\in
V_{G}$. Informally, $X_{i,v}$ should express that vertex $v$ is visited in
the step $i$ in a path on $G$. Define propositional formulas $A-E$ as
follows and let $\alpha _{G}:=A\wedge B\wedge C\wedge D\wedge E$.

\begin{enumerate}
\item  \label{A} $A=\bigwedge_{v\in V}\left( X_{1,v}\vee \ldots \vee
X_{n,v}\right) $ (: every vertex is visited in $X$).

\item  \label{B} $B=\bigwedge_{v\in V}\bigwedge_{i\neq j}\left(
X_{i,v}\rightarrow \left( X_{j,v}\rightarrow \bot \right) \right) $ (: there
are no repetitions in $X$).

\item  \label{C} $C=\bigwedge_{i\in \left[ n\right] }\bigvee_{v\in V}X_{i,v}$
(: at each step at least one vertex is visited).

\item  \label{D} $D=\bigwedge_{v\neq w}\bigwedge_{i\in \left[ n\right]
}\left( X_{i,v}\rightarrow \left( X_{i,w}\rightarrow \bot \right) \right) $
(: at each step at most one vertex is visited).

\item  \label{E} $E=\bigwedge_{(v,w)\not\in E}\bigwedge_{i\in \left[ n-1%
\right] }\left( X_{i,v}\rightarrow \left( X_{i+1,w}\rightarrow \bot \right)
\right) $ (: if there is no edge from $v$ to $w$ then $w$ can't be visited
immediately after $v$).
\end{enumerate}
\end{definition}

Thus $G$ is Hamiltonian iff $\alpha _{G}$ is satisfiable. Denote by $%
SAT_{Cla}$ the set of satisfiable formulas in classical propositional logic
and by $TAUT_{Int}$ the set of tautologies in the intuitionistic one. Then
the following conditions hold: (1) $G$ is non-Hamiltonian iff $\alpha
_{G}\not\in SAT_{Cla}$, (2) $G$ is non-Hamiltonian iff $\lnot \alpha _{G}\in
TAUT_{Cla}$, (3) $G$ is non-Hamiltonian iff $\lnot \alpha _{G}\in TAUT_{Int}$%
. Glyvenko's theorem yields the equivalence between (2) to (3). Hence $G$ is
non-Hamiltonian iff there is an intuitionistic proof of $\lnot \alpha _{G}$.
Such proof is called a certificate for the non-hamiltoniancy of $G$. \cite
{Statman} (also \cite{Haeus}) 9presented a translation from formulas in full
propositional intuitionistic language into the purely implicational fragment
of minimal logic whose formulas are built up from $\rightarrow $ and
propositional variables. This translation employs new propositional
variables $q_{\gamma }$ for logical constants and complex propositional
formulas $\gamma $ (in particular, every $\alpha \vee \beta $ and $\alpha
\wedge \beta $ should be replaced by $q_{\alpha \vee \beta }$ and $q_{\alpha
\wedge \beta }$, respectively) while adding implicational axioms stating
that $q_{\gamma }$ is equivalent to $\gamma $ . For any propositional
formula $\gamma $, let $\gamma ^{\star }$ denote its translation into purely
implicational minimal logic in question. Note that $size\left( \gamma
^{\star }\right) \leq size^{3}\left( \gamma \right) $. Now $\gamma \in
TAUT_{Int}$ iff $\gamma ^{\star }$ is provable in the minimal logic.
Moreover, it follows from \cite{Statman}, \cite{Haeus} that for any normal
ND proof $\partial $ of $\gamma $ there is a normal proof $\partial
_{\rightarrow }$ of $\gamma ^{\star }$ in the corresponding ND system for
minimal logic, \textsc{NM}$_{\rightarrow }$, such that $h\left( \partial
_{\rightarrow }\right) =\mathcal{O}\left( h\left( \partial \right) \right) $%
. Thus in order to prove Lemma 12 it will suffice to establish

\begin{claim}
$G$ is non-Hamiltonian iff there exists a normal intuitionistic tree-like ND
proof of $\alpha _{G}\rightarrow \bot $ , i.e. $\lnot \alpha _{G}$, whose
height is polynomial in $n$.
\end{claim}

\begin{proof}
Straightforward (see \cite{GH3} for details).
\end{proof}

This completes proofs of Lemma 12 and Corollary 13.

\subsubsection{Case $\mathbf{NP}$ vs $\mathbf{PSPACE}$}

In the sequel we consider standard language $\mathcal{L}_{\rightarrow }$ of
minimal logic whose formulas ($\alpha $, $\beta $, $\gamma $, $\rho $ etc.)
are built up from propositional variables ($p$, $q$, $r$, etc.) using one
propositional connective `$\rightarrow $'. The sequents are in the form $%
\Gamma \Rightarrow \alpha $\ whose antecedents, $\Gamma $,\ are viewed as
multisets of formulas; sequents $\Rightarrow \alpha $\ , i.e. $\emptyset
\Rightarrow \alpha $, are identified with formulas $\alpha $.

Recall that HSC for minimal logic, \textsc{LM}$_{\rightarrow }$, includes
the following axioms $\left( \text{\textsc{M}}A\right) $ and inference rules 
$\left( \text{\textsc{M}}I1\rightarrow \right) $, $\left( \text{\textsc{M}}%
I2\rightarrow \right) $, $\left( \text{\textsc{M}}E\rightarrow P\right) $, $%
\left( \text{\textsc{M}}E\rightarrow \rightarrow \right) $ in the language $%
\mathcal{L}_{\rightarrow }$ (the constraints are shown in square brackets). 
\footnote{{\footnotesize This is a slightly modified, equivalent version of
the corresponding \ purely implicational and }$\bot ${\footnotesize -free} 
{\footnotesize subsystem of Hudelmaier's intuitionistic calculus \textsc{LG}%
, cf. \cite{Hud}. The constraints }$q\in VAR\left( \Gamma ,\gamma \right) $%
{\footnotesize \ are added just for the sake of transparency.}}\medskip 

$\fbox{$\left( \text{\textsc{M}}A\right) :\ \ \ \Gamma ,p\Rightarrow p$}$

$\fbox{$\left( \text{\textsc{M}}I1\!\rightarrow \right) :\ \ \ \dfrac{\Gamma
,\alpha \Rightarrow \beta }{\Gamma \Rightarrow \alpha \rightarrow \beta }%
\smallskip \quad \left[ \left( \nexists \gamma \right) :\left( \alpha
\rightarrow \beta \right) \rightarrow \gamma \in \Gamma \right] $}$

$\fbox{$\left( \text{\textsc{M}}I2\!\rightarrow \right) :\ \ \ \dfrac{\Gamma
,\alpha ,\beta \rightarrow \gamma \Rightarrow \beta }{\Gamma ,\left( \alpha
\rightarrow \beta \right) \rightarrow \gamma \Rightarrow \alpha \rightarrow
\beta }$}$

$\fbox{$\left( \text{\textsc{M}}E\!\rightarrow \!P\right) :\ \ \ \dfrac{%
\Gamma ,p,\gamma \Rightarrow q}{\Gamma ,p,p\rightarrow \gamma \Rightarrow q}%
\quad \left[ q\in \mathrm{VAR}\left( \Gamma ,\gamma \right) ,p\neq q\right] $%
}$

$\fbox{$\left( \text{\textsc{M}}E\!\rightarrow \rightarrow \right) :\ \ \ 
\dfrac{\Gamma ,\alpha ,\beta \rightarrow \gamma \Rightarrow \beta \quad
\quad \Gamma ,\gamma \Rightarrow q}{\Gamma ,\left( \alpha \rightarrow \beta
\right) \rightarrow \gamma \Rightarrow q}\quad \left[ q\in \mathrm{VAR}%
\left( \Gamma ,\gamma \right) \right] $}$

\begin{claim}
\textsc{LM}$_{\rightarrow }$ is sound and complete with respect to minimal
propositional logic and tree-like deducibility. Any given formula $\rho $ is
valid in the minimal logic iff sequent $\Rightarrow \rho $ is provable in 
\textsc{LM}$_{\rightarrow }$ by a quasi-polynomial tree-like deduction.
\end{claim}

\begin{proof}
Easily follows from \cite{Hud} (see \cite{GH1} for details).
\end{proof}

\begin{lemma}
For any valid purely implicational formula $\rho $ there exists a
quasi-polynomial tree-like proof $\partial \vdash \rho $ in \textsc{NM}$%
_{\rightarrow }$.
\end{lemma}

\begin{proof}
This $\partial $ is a straightforward interpretation in \textsc{NM}$%
_{\rightarrow }$ of a proof in \textsc{LM}$_{\rightarrow }$ that must exist
by the validity of $\rho $ (see \cite{GH1} for details).
\end{proof}

Recall that the validity problem in minimal logic is $\mathbf{PSPACE}$%
-complete \cite{Statman}, \cite{Svejdar}. Together with Theorem 5 and
Corollary 12 this yields

\begin{corollary}
$\mathbf{NP=PSPACE}$\textbf{\ }holds true.
\end{corollary}

\begin{corollary}
The satisfiability and validity problems in quantified boolean logic (QBL)
are both $\mathbf{NP}$-complete, since corresponding $\mathbf{PSPACE}$%
-completeness is well-known (see e.g. \cite{arora}, \cite{Papa}). Moreover $%
\mathbf{BQP\subseteq NP}$ holds, where $\mathbf{BQP}$ is the class of
problems computable in quantum polynomial time. This follows from the known
inclusion $\mathbf{BQP\subseteq PSPACE}$ (cf. \cite{arora}).
\end{corollary}

\begin{conclusion}[$\mathbf{PSPACE}$ paradise]
Denote by $\mathbb{U}$ the universe of solvable computational problems and
let $\mathbb{V:}=\mathbf{PSPACE\varsubsetneq \,}\mathbb{U}$ be the proper
subuniverse consisting of problems solvable in polynomial space.

Thus $\mathbb{V}$ contains all problems whose solutions are polynomially
admissible with respect to the space used (regardless of the required time).
Loosely speaking, $\mathbb{V}$ is the world of problems solvable in the
material world of sufficiently big computers, without any time restriction.
It is known that $\mathbb{V}$ preserves basic propositional operations and
non-deterministic provability and includes $\mathbf{BQP}$ (cf. e.g. \cite
{arora}, \cite{Papa}). 

Let $\mathbb{W:}=\mathbf{NP\subseteq \,}\mathbb{V\,}\mathbf{\varsubsetneq \,}%
\mathbb{U}$ be another subuniverse consisting of problems that are
potentially solvable in polynomial time. Now Corollary 18 shows that $%
\mathbb{W}=\mathbb{V}$, i.e. any given problem $X\in \mathbb{V}$ (in
particular $X\in \mathbf{BQP}$) is in fact solvable by some deterministic
polynomial-time TM $M_{X}$. Hence all problems in $\mathbb{V}$ are
polynomially admissible with respect to both space and time used. To
paraphrase Hilbert's famous quotation: \emph{in }$\mathbb{V}$\emph{\ there
is no polynomial-time ignorabimus}. One can ask whether $M_{X}$ can be
obtained from $X$ by a polynomial-time algorithm. The answer in NO, provided
that $\mathbf{P}\neq \mathbf{NP}$.
\end{conclusion}

--------------------------------------------------------------------------------------------

--------------------------------------------------------------------------------------------

\end{document}